\documentclass[11pt,letterpaper]{article}
\usepackage{fullpage}
\usepackage{amsmath,amssymb,amsthm,mathtools,microtype}
\usepackage{tikz}
\usepackage[hidelinks]{hyperref}
\newtheorem{theorem}{Theorem}[section]
\newtheorem{lemma}[theorem]{Lemma}
\newtheorem{claim}[theorem]{Claim}
\newtheorem{corollary}[theorem]{Corollary}
\theoremstyle{definition}
\newtheorem{definition}[theorem]{Definition}
\theoremstyle{remark}
\newtheorem{remark}[theorem]{Remark}
\newtheorem{observation}[theorem]{Observation}
\newcommand{\Q}{\mathbb Q}
\hypersetup{pdftitle={Matrix Multiplication from Polynomial Packings: A Constructive and Elementary Proof of omega <= 9/4},
pdfauthor={Rafael Pass},
pdfsubject={Polynomial packing and an explicit matrix multiplication algorithm}}
\title{Matrix Multiplication from Polynomial Packings:\\
A Constructive Proof of $\omega\le9/4$}
\author{Rafael Pass\thanks{Pearl Research, Cornell Tech, Technion, TAU}}
\date{October 30, 2026}
\begin{document}
\hypersetup{pageanchor=false}
\pagenumbering{gobble}
\maketitle
\thispagestyle{empty}
\begin{abstract}
We give a self-contained constructive proof that the
matrix multiplication exponent satisfies $\omega\le9/4$. In a recent
breakthrough, OpenAI~\cite{original} established this bound in the model
of exact scalar arithmetic over the complex numbers, using
non-constructive arguments in Strassen's asymptotic spectral framework.
We explicitly construct a Strassen-style matrix multiplication identity
and obtain the bound by ordinary block recursion. For every
$\gamma>0$, the resulting algorithm multiplies two $n\times n$
integer matrices using $O(n^{9/4+\gamma})$ integer arithmetic
operations.

As an independent contribution, we  develop a framework of
\emph{polynomial packings}---methods for reducing many polynomial
multiplications to fewer ones using linear operations. We isolate
and generalize a construction of OpenAI~\cite{original} to obtain
a simple packing, and show how packings can be composed to amplify
their ``compression rate", which finally yields Strassen-like identities.
\end{abstract}

\clearpage
\pagenumbering{arabic}
\hypersetup{pageanchor=true}
\section{Introduction}

The exponent $\omega$ of matrix multiplication is the infimum of all
$\tau$ such that two $n\times n$ matrices can be multiplied with
$O(n^\tau)$ arithmetic operations. Strassen~\cite{strassen69} showed
$\omega\le\log_2 7<2.81$. Later bounds follow a scheme developed by
Bini, Capovani, Romani, and Lotti~\cite{bclr79,bini},
Sch\"onhage~\cite{schonhage81}, and Strassen~\cite{strassen87}.
An explicit identity computes several products together, possibly with
an error that vanishes with a parameter. Combining identities amplifies
the saving, and the resulting independent products give a bound on
$\omega$. Refinements of Strassen's laser method, starting from the
identity of Coppersmith and Winograd~\cite{cw90}, reached
$\omega<2.371339$~\cite{alman25} and
$\omega<2.371177$~\cite{dupont26}.

In a recent breakthrough, OpenAI showed $\omega\le9/4$~\cite{original},
following a companion bound
$\omega<2.258$~\cite{companion}, starting from identities for polynomial
multiplication. Their theorem is stated for \emph{exact arithmetic}
over the complex numbers, charging unit cost per scalar operation.
Their proof uses Strassen's asymptotic
spectrum~\cite{strassen88}, an algebraic and topological framework for
comparing bilinear maps through numerical invariants. Concavity
strengthens inequalities for these invariants; spectral duality then
converts them into the exponent bound. Their argument establishes the
exponent bound without explicitly constructing the corresponding algorithms.

We give an explicit, constructive algorithm reaching the same bound
using integer arithmetic throughout; furthermore, our construction
is self-contained. Our starting point is the three-for-one polynomial
degeneration of~\cite[Lemma~4.2]{original}. We generalize it to a
$(c+2)$-for-$c$ polynomial packing and next use our $(c+2)$-for-$c$ gadget, together with
tensoring, to construct a Strassen-style identity: a fixed
$v\times v$ block matrix product can be computed using $r$ products
of linear combinations of its blocks, with $\log_v r$ arbitrarily
close to $9/4$. This replaces the usual $v^3$ block products by $r$,
just as Strassen replaces eight by seven. Applying the fixed identity
recursively then gives the exponent bound.

\begin{theorem}\label{thm:main}
For every $\gamma>0$, there is an explicit recursive algorithm that,
for every positive integer $n$, multiplies two $n\times n$ integer
matrices with $O(n^{9/4+\gamma})$ integer arithmetic operations.
In particular, $\omega\le9/4$.
\end{theorem}

The algorithm remains ``galactic'': the constants hidden in the
$O(\cdot)$ notation are enormous.

\subsection{Our techniques}

Our approach develops a framework of \emph{polynomial packings}:
methods for reducing many polynomial multiplications to fewer ones
using linear operations. We first isolate and generalize a construction
of OpenAI~\cite{original} to obtain a simple packing; we next show how
packings can be composed to amplify their ``compression rate", and
finally how they yield Strassen-like identities achieving the desired
exponent bound.

\paragraph{Polynomial packing.}
Informally, given $m$ pairs of polynomials $(A_i,B_i)$, an $m$-for-$n$
polynomial packing computes $C_i=A_iB_i$ for all $i$, or perhaps a
\emph{selected subset} of their coefficients using $n$ polynomial
multiplications of encoded inputs, with only linear operations for
encoding and recovery. We keep track of both the number of
multiplications and the sizes of the polynomials involved.

We first show how to construct a so-called \emph{same-left (SL)}
$(c+2)$-for-$c$ packing modulo a formal parameter
$\varepsilon$\footnote{Here $\varepsilon$ is an additional symbolic
variable. In other words, we use linear maps whose coefficients are
polynomials in $\varepsilon$ and $\varepsilon^{-1}$. After decoding,
the output is polynomial in $\varepsilon$, with the desired answer
as its constant term. Bini's standard interpolation argument~\cite{bini}
removes the parameter at the end; see Lemma~\ref{lem:bini}.},
generalizing the three-for-one construction of~\cite[Lemma~4.2]{original}.
Same-left here means that all the left polynomials $A_i$ are the same.
More precisely, consider univariate polynomials $A,P_1,\ldots,P_{c+1}$ with scalar
coefficients, where $A$ has length $a$ (degree less than $a$)
and each $P_i$ has length $b$.
Write $[z^k]H$ for the coefficient of $z^k$ in $H$,
taking coefficients outside its range to be zero.
We want to compute all coefficients of the $c+1$ products
$AP_1,\ldots,AP_{c+1}$, which we call \emph{full products}.
We denote the full-product operation with input lengths $a,b$ by
$\operatorname{Full}_{a,b}$, or simply Full.
Naively, this takes $c+1$ polynomial multiplications.
We can trivially reduce the number of polynomial multiplications to $c$ by replacing
$P_c,P_{c+1}$ with $P_c+z^{a+b-1}P_{c+1}$: the two products then
occupy disjoint coefficient ranges and can be read off separately.
But this encoded input has length $a+2b-1$, at least the combined
length $2b$ of the two original inputs. The smaller number of polynomial multiplications
therefore gives no saving in the total number of input coefficients.

Alternatively, we can use $c$ encoded right polynomials with just one
extra coefficient each. For $c\ge a+b-1$, set
\[
 E_j^{(0)}=P_j-zP_{j+1},\qquad 1\le j\le c.
\]
Each $E_j^{(0)}$ has length $b+1$. After multiplying these polynomials
by $A$, take the following telescoping sum:
\[
 \begin{aligned}
 \sum_{j=1}^c z^{j-1}AE_j^{(0)}
 &=(AP_1-zAP_2)+(zAP_2-z^2AP_3)+\cdots
   +(z^{c-1}AP_c-z^cAP_{c+1})\\
 &=AP_1-z^cAP_{c+1}.
 \end{aligned}
\]
This sum is computed by adding coefficients. Writing
$W_j=AE_j^{(0)}$ and $W$ for the displayed sum, we have
$[z^k]W=\sum_{j=1}^c[z^{k-j+1}]W_j$.
Thus forming $W$ is a linear decoding operation.
The intermediate products cancel, so the first $a+b-1$ coefficients
give $AP_1$, and the differences $AE_j^{(0)}=AP_j-zAP_{j+1}$ then
recover the remaining products successively. This still gives no saving
in the total number of right-input coefficients: $c(b+1)\ge(c+1)b$.

But notice the \emph{empty space} in the combined output: $AP_1$ ends before
position $a+b-1$, while $z^cAP_{c+1}$ starts at position $c$.
Taking $c\ge a+2b-1$ leaves at least $b$ zero coefficient positions
between them. We use this space to obtain an additional product output,
without making the encoded polynomials any longer.

So let us introduce an additional right-input polynomial $Q$ of length
$a+b-1$, aiming for a $(c+2)$-for-$c$ packing, and write
$Q=\sum_{j=0}^{c-1}q_jz^j$.

\begin{figure}[htbp]
\centering
\begin{tikzpicture}[x=1cm,y=1cm,font=\small]
 \definecolor{gapblue}{RGB}{34,100,145}
 \node[anchor=west] at (0,3.15) {Without $Q$};
 \filldraw[fill=black!6,draw=black!55] (0,2.25) rectangle (4.2,2.95);
 \filldraw[fill=black!2,draw=black!55,dashed] (4.2,2.25) rectangle (7.2,2.95);
 \filldraw[fill=black!6,draw=black!55] (7.2,2.25) rectangle (11.4,2.95);
 \node at (2.1,2.6) {$AP_1$};
 \node at (5.7,2.6) {zero coefficients};
 \node at (9.3,2.6) {$-z^cAP_{c+1}$};

 \node[anchor=west] at (0,1.65) {With $\varepsilon Q$};
 \filldraw[fill=black!6,draw=black!55] (0,.65) rectangle (4.2,1.45);
 \filldraw[fill=gapblue!13,draw=gapblue,thick] (4.2,.65) rectangle (7.2,1.45);
 \filldraw[fill=black!6,draw=black!55] (7.2,.65) rectangle (11.4,1.45);
 \node[align=center] at (2.1,1.05) {overlapping\\contributions};
 \node[align=center,text=gapblue] at (5.7,1.05)
   {middle of $\varepsilon AQ$\\shifted by $z^b$};
 \node[align=center] at (9.3,1.05) {overlapping\\contributions};

 \draw[densely dotted,black!45] (4.2,2.25) -- (4.2,1.45);
 \draw[densely dotted,black!45] (7.2,2.25) -- (7.2,1.45);
 \draw[densely dotted,black!45] (4.2,.65) -- (4.2,0);
 \draw[densely dotted,black!45] (7.2,.65) -- (7.2,0);
 \draw[->,black!65] (0,0) -- (11.8,0);
 \foreach \x in {0,4.2,7.2} \draw (\x,.07) -- (\x,-.07);
 \node[below] at (0,-.08) {$0$};
 \node[below] at (4.2,-.08) {$a+b-1$};
 \node[below] at (7.2,-.08) {$c=a+2b-1$};
 \node[anchor=north east,font=\footnotesize] at (11.8,-.55) {coefficient position};
\end{tikzpicture}
\caption{The telescoping sum before and after adding $\varepsilon Q$, shown for
$c=a+2b-1$. The $b$ highlighted coefficients of $\varepsilon z^bAQ$ fill the zero gap;
its remaining coefficients can overlap the endpoint products.}
\label{fig:telescoping-gap}
\end{figure}
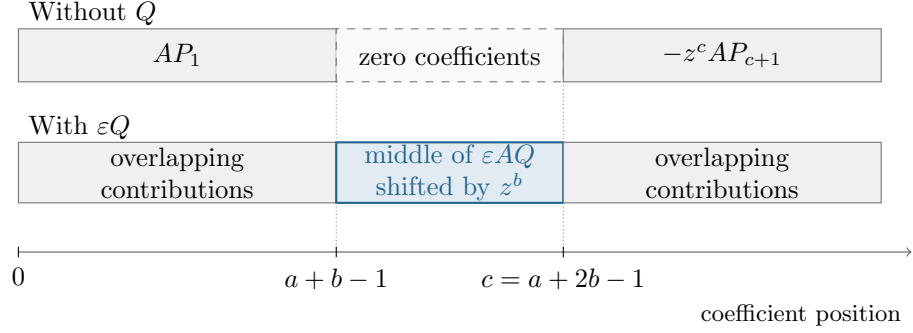
We now make use of the formal parameter $\varepsilon$ and modify
the encoding to:
\[
 E_j=P_j-zP_{j+1}+\varepsilon z^b q_{j-1},\qquad 1\le j\le c.
\]
\begin{quote}
\emph{That is, we add just one coefficient of $\varepsilon Q$ to each $E_j$, at position $b$.
Each $P_j$ ends at position $b-1$, but the shifted $zP_{j+1}$
already reaches position $b$: the added coefficient overlaps with
its last coefficient.}
\end{quote}
The same telescoping sum now gives:
\[
 \sum_{j=1}^c z^{j-1}AE_j
 =AP_1-z^cAP_{c+1}+\varepsilon z^bAQ.
\]
The coefficients at positions $a-1,\ldots,a+b-2$ of
$\varepsilon AQ$, shifted by $z^b$, lie entirely in the gap.
Dividing them by $\varepsilon$ recovers these coefficients of
$AQ$ exactly; we call this operation the \emph{middle product},
denoted $\operatorname{Mid}_{a,b}$, or simply Mid.\footnote{\label{footnote}In the formal
proof and the tensor formulas below, we shift the right-input and
output exponents of the middle product by $-(a+b-2)$. This only
relabels coefficients and leaves the left input unchanged.}
For the full products, use the same successive decoder as before
and reduce modulo $\varepsilon$: every contribution from $Q$
vanishes. Our so-called \emph{same-left primal packing} therefore
computes
\[
 \underbrace{(AP_i)_{i=1}^{c+1}}_{c+1\text{ full products}}
 \qquad
 \underbrace{\bigl([z^k](AQ)\bigr)_{k=a-1}^{a+b-2}}
 _{1\text{ middle product}}
 \qquad\pmod{\varepsilon}.
\]
Reduction modulo $\varepsilon$ means taking the constant coefficient
of the fully decoded result. The only division by $\varepsilon$
is applied to the middle output, which is divisible by
$\varepsilon$, so no negative powers remain.
Thus \emph{we compute $c+2$ polynomial-product outputs modulo
$\varepsilon$ using $c$ multiplications, with only one extra
coefficient in each encoded right polynomial}.

The $(c+1)$-for-$c$ saving alone is offset by the increased input
lengths; the additional middle product provides the improvement.
(Introducing $\varepsilon$ also has a cost, but its degree is only
one here. We use a single formal parameter throughout, which
Bini's standard technique (Lemma~\ref{lem:bini}) removes at the end
without affecting the exponent bound. We can thus ignore this
cost in the discussion below.)

More concretely, to see the gain and its cost, take $a=b=L\ge3$ and $c=3L$.
The encoded polynomials then have length $L+1=(1+3/c)L$.
In return, $c$ source products provide $c+1$ full outputs and one
middle output, giving an output-count gain of $(c+2)/c=1+2/c$.
(As we shall see below, we can amplify this gain multiplicatively
by applying the same method in multiple variables.)

\paragraph{Exchanging the roles in a packing.}
A polynomial product is linear in either input when the other is fixed.
Write its \emph{shape} as $P=(\mathcal I,\mathcal J,\mathcal K)$,
the allowed monomials in the left input, right input, and output,
which we call roles $1,2,3$. Exchanging two roles gives
\[
 \begin{aligned}
 P^{12}&=(\mathcal J,\mathcal I,\mathcal K),\\
 P^{23}&=(\mathcal I,\mathcal K^{-1},\mathcal J^{-1}),\\
 P^{13}&=(\mathcal K^{-1},\mathcal J,\mathcal I^{-1}),
 \end{aligned}
\]
where $\mathcal S^{-1}=\{u^{-1}:u\in\mathcal S\}$ inverts monomial
labels, so $z^k$ becomes $z^{-k}$. These inverses preserve the
multiplication relation: $uv=w$ becomes $uw^{-1}=v^{-1}$ or
$w^{-1}v=u^{-1}$.

For a packing with independent inputs, applying the same exchange
to every target and source gives another such packing with the same
counts and the same exactness
or validity modulo $\varepsilon$. The exchange $(12)$ swaps the
input roles; the input--output exchanges follow from the standard
transposition principle~\cite[Chapter~13]{bcs97}.
The exchange $(23)$ also preserves the same-left requirement because
it keeps the left-input role fixed. We call this exchange
\emph{duality}, write $\widehat P=P^{23}$, and recover the original
product by applying it twice.
We next observe that\footnote{Formally, this relies on the coefficient shift noted in footnote \ref{footnote}.}
\[
 \operatorname{Full}_{a,b}^{23}=\operatorname{Mid}_{a,b}.
\]
Thus our primal packing gives a dual same-left packing, modulo
$\varepsilon$, of $c+1$ middle products and one full product at
lengths $a,b$, using $c$ middle products at lengths $a,b+1$.

\paragraph{Amplifying the gain: Multivariate Packing.}
Running the primal and dual packings side by side gives $c+2$ full
and $c+2$ middle outputs from $c$ full and $c$ middle source products.
Their counts add, so the gain remains $2(c+2)/(2c)=(c+2)/c$.
To multiply the gain, we move to multivariate polynomials and apply
the packings in separate variables. In two variables $x,y$, one
primal step in $x$ and one dual step in $y$ give $(c+2)^2$ outputs
from $c^2$ source products;
continuing in additional variables multiplies the gain further.
To see why, start with a common left polynomial $A(x,y)$ of length
$a$ in each variable and $(c+2)^2$ independent right polynomials
$B_1(x,y),B_2(x,y),\ldots$. Arrange the right polynomials in a
$(c+2)\times(c+2)$ table, writing $B_{i,j}$ for the polynomial
in row $i$ and column $j$. We want coefficients of each product
$AB_{i,j}$. For rows $1,\ldots,c+1$, we want Full in $x$;
for the last row, Mid in $x$. For columns $1,\ldots,c+1$, we
want Mid in $y$; for the last column, Full in $y$.
Each $B_{i,j}$ has length $b$ in a variable where we want Full
and length $a+b-1$ where we want Mid.
Figure~\ref{fig:bivariate-table} shows every row and column for $c=4$.
\begin{figure}[htbp]
\centering
\begin{tikzpicture}[x=.63cm,y=.60cm,font=\scriptsize]
\definecolor{pairFM}{RGB}{226,239,249}
\definecolor{pairFF}{RGB}{253,236,211}
\definecolor{pairMM}{RGB}{229,241,224}
\definecolor{pairMF}{RGB}{239,228,245}
\node[font=\small] at (3,7.9) {Right inputs $B_{i,j}$};
\node[font=\small] at (13,7.9) {Encoded inputs $U_{r,j}$};
\node[font=\small] at (22,7.9) {Encoded inputs $V_{r,s}$};
\node at (3,7.2) {$6\times6$};
\node at (13,7.2) {$4\times6$};
\node at (22,7.2) {$4\times4$};
\foreach \j in {1,...,6} {
  \pgfmathsetmacro{\xx}{\j-.5}
  \ifnum\j<6
    \node at (\xx,6.5) {Mid};
    \node at ({10+\xx},6.5) {Mid};
  \else
    \node at (\xx,6.5) {Full};
    \node at ({10+\xx},6.5) {Full};
  \fi
}
\foreach \i in {1,...,6} {
  \ifnum\i<6
    \node[anchor=east] at (-.15,{6.5-\i}) {Full};
  \else
    \node[anchor=east] at (-.15,{6.5-\i}) {Mid};
  \fi
  \foreach \j in {1,...,6} {
    \ifnum\i<6
      \ifnum\j<6 \def\cellcolor{pairFM}\else\def\cellcolor{pairFF}\fi
    \else
      \ifnum\j<6 \def\cellcolor{pairMM}\else\def\cellcolor{pairMF}\fi
    \fi
    \filldraw[fill=\cellcolor,draw=gray!65] ({\j-1},{6-\i}) rectangle (\j,{7-\i});
    \node at ({\j-.5},{6.5-\i}) {$B_{\i,\j}$};
  }
}
\foreach \r in {1,...,4} {
  \node[anchor=east] at (9.85,{6.5-\r}) {Full};
  \node[anchor=east] at (19.85,{6.5-\r}) {Full};
  \node at ({19.5+\r},6.5) {Mid};
  \foreach \j in {1,...,6} {
    \ifnum\j<6 \def\cellcolor{pairFM}\else\def\cellcolor{pairFF}\fi
    \filldraw[fill=\cellcolor,draw=gray!65] ({9+\j},{6-\r}) rectangle ({10+\j},{7-\r});
    \node at ({9.5+\j},{6.5-\r}) {$U_{\r,\j}$};
  }
  \foreach \ss in {1,...,4} {
    \filldraw[fill=pairFM,draw=gray!65] ({19+\ss},{6-\r}) rectangle ({20+\ss},{7-\r});
    \node at ({19.5+\ss},{6.5-\r}) {$V_{\r,\ss}$};
  }
}
\draw[->,thick] (6.25,4) -- (8.65,4)
 node[midway,above,align=center,font=\tiny] {Primal in $x$\\each column};
\draw[->,thick] (16.25,4) -- (18.65,4)
 node[midway,above,align=center,font=\tiny] {Dual in $y$\\each row};
\node[align=center] at (3,-.8) {Rows: requirement in $x$\\Columns: requirement in $y$};
\node[align=center] at (13,1) {Six inputs become four\\in each column};
\node[align=center] at (22,1) {Six inputs become four\\in each row};
\end{tikzpicture}
\caption{The bivariate construction for $c=4$. Each original column
contains five Full inputs and one Mid input in $x$; each original
row contains five Mid inputs and one Full input in $y$.
After encoding, compute the sixteen products $AV_{r,s}$, Full in
$x$ and Mid in $y$, at the longer source lengths. Decode in $y$
and then in $x$ to recover all thirty-six requested outputs modulo
$\varepsilon$.}
\label{fig:bivariate-table}
\end{figure}
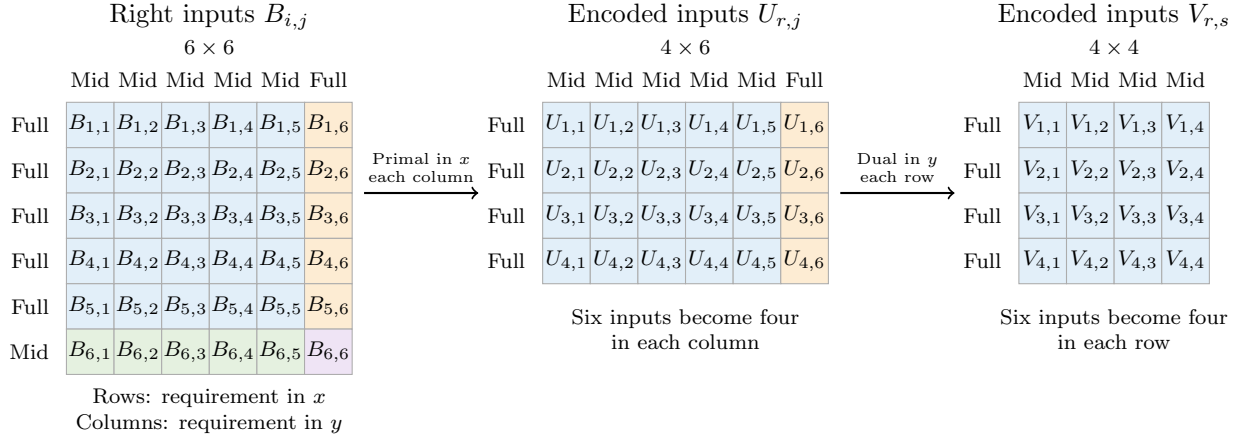
For example, each of the $(c+1)^2$ upper-left products requests
$[x^k y^\ell](AB_{i,j})$ for $0\le k<a+b-1$ and
$a-1\le\ell<a+b-1$; its right input has lengths $(b,a+b-1)$.
Every column supplies the $c+1$ full-range inputs and one
middle-range input needed by the primal packing in $x$.
Every row supplies the $c+1$ middle-range inputs and one
full-range input needed by the dual packing in $y$.

As above, take $c\ge a+2b-1$. In each column, apply the primal
encoder in $x$, viewing the coefficients as polynomials in $y$.
Within each column, all inputs have the same $y$-length bound.
The encoder takes linear combinations of their $x$-coefficients
with weights independent of $y$, so this bound is preserved.
This leaves $c$ rows, each still containing the required $c+1$
Mid inputs and one Full input in $y$. In each resulting row,
apply the dual encoder in $y$, viewing the
coefficients as polynomials in $x$. We now have a $c\times c$
table of encoded right polynomials, each of length $b+1$ in $x$
and $a+b$ in $y$.
Multiply each by $A$, computing all $x$-coefficients and the
$y$-coefficients at exponents $a-1,\ldots,a+b-1$. Apply the dual
decoder in $y$, then the primal decoder in $x$.

Encoding in $x$ and then in $y$, computing the source products,
and decoding in reverse order recovers all the requested outputs
modulo $\varepsilon$. This follows by applying the correctness of
the two univariate packings along their separate coefficient indices.

We therefore recover $(c+2)^2$ outputs from $c^2$ source products.
Continuing this process in additional variables raises the
univariate compression rate to higher powers. Use $2d$ variables
$z_1,\ldots,z_{2d}$, applying the primal packing in the first $d$
and the dual packing in the last $d$. Encode in each variable,
compute the resulting products, and decode in reverse order. This gives
$(c+2)^{2d}$ outputs from $c^{2d}$ source products, an output-to-source
count ratio of $((c+2)/c)^{2d}$. Accounting also for the longer
source polynomials gives the compression rate defined below.

However, there is one issue: after the above operations, the source products are all Full in the first $d$
variables and Mid in the last $d$, whereas the recovered
outputs can have different numbers of variables with Full and Mid requirements.
Write $P\otimes Q$ for combining two product shapes on separate
variable lists, and $P^{\otimes d}$ for $d$ copies.
Put
\[
 T_{a,b}
 =(\operatorname{Full}_{a,b}\otimes\operatorname{Mid}_{a,b})^{\otimes d}.
\]
The sources have this form at the longer source lengths. To ensure that source computations and target computations have the same ``shape", only  keep
targets with Full in exactly $d$ variables and Mid in the other $d$, so their shapes
agree with $T_{a,b}$ up to renaming variables. We refer to such products
as \emph{balanced}.
For example, in the two-variable table, the yellow and green
outputs require Full in both variables or Mid in both, respectively, so
they do not qualify. The $(c+1)^2$ blue outputs and the single
purple output have one of each: $26$ of the $36$ outputs when
$c=4$.
With $2d$ variables, we recover at least
\[
 \frac{(c+2)^{2d}}{d+1}
\]
outputs of this shape.\footnote{To see why, consider the following
equivalent counting question. Form a word whose first $d$ letters
are chosen from $\{\mathsf F_1,\ldots,\mathsf F_{c+1},\mathsf M\}$
and whose last $d$ letters are chosen from
$\{\mathsf F,\mathsf M_1,\ldots,\mathsf M_{c+1}\}$.
All $\mathsf F$-letters have type Full and all $\mathsf M$-letters
have type Mid; subscripts distinguish different letters of the same type.
How many words contain exactly $d$ letters of each type?
Cauchy--Schwarz gives at least $(c+2)^{2d}/(d+1)$ such words;
see Claim~\ref{lem:count-words}.}
Thus requiring this common shape costs at most a factor $d+1$
in compression. 
(Note that even in a balanced product, the Full and Mid requirements may be assigned to different
variables; the next step lets us deal with this.)

\paragraph{Removing the same-left restriction.}
So far, the same-left polynomial packing (SL-PP) computes $w$ retained
targets sharing one left input $A$. Index them by $i=0,\ldots,w-1$.
We now want independent pairs $A_i,B_i$. Write each right input as
\[
 B_i=[\,B_{i,0}\ \cdots\ B_{i,w-1}\,],
\]
where each $B_{i,j}$ is a scalar polynomial and the brackets concatenate
them into a polynomial with row-vector coefficients of length $w$.
The entries of each
$B_i$ are all multiplied by the same scalar polynomial $A_i$.
We batch the $A_i$ into one auxiliary univariate polynomial
$A(\xi)$ satisfying $A(i)=A_i$. For each $i$, encode the column
blocks $B_{i,j}$ in a univariate polynomial $B_i'(\xi)$ satisfying
$B_i'(i+\varepsilon j)=B_{i,j}$. Both encodings have degree at most
$w-1$, so the requested coefficients of $A(\xi)B_i'(\xi)$ are
polynomials of degree at most $2w-2$ in $\xi$.
We recover these polynomials by interpolation from $2w-1$ common
evaluation points. Fix one such point $k$. We need the requested
coefficients of the products
$A(k)B_0'(k),\ldots,A(k)B_{w-1}'(k)$.
Their right inputs differ, but their left input is the single value
$A(k)$ of the polynomial encoding all the $A_i$. Thus one application
of the same-left packing computes all these products' requested
coefficients modulo $\varepsilon$. We repeat this for each of the
$2w-1$ evaluation points; $A(k)$ may change between applications.

The points $i+\varepsilon j$ were chosen for two reasons.
They are distinct, allowing $B'_i$ to take the independently
prescribed values $B_{i,j}$, but they are all congruent to $i$
modulo $\varepsilon$. Thus they distinguish the right inputs
while keeping the left input equal to $A_i$ modulo $\varepsilon$,
since $A(i)=A_i$.
Evaluating the recovered coefficient polynomials at these points
therefore gives the requested coefficients of $A_iB_{i,j}$
modulo $\varepsilon$.\footnote{The formal proof replaces the
original packing's parameter by $\varepsilon^w$, ensuring that
divisions by $\varepsilon$ in the right-input encoding leave
the error divisible by $\varepsilon$.}
Concatenating these output blocks recovers the requested
coefficients of $A_iB_i$ for each $i$.
(OpenAI~\cite{original} proves a related result using finite
Fourier projection followed by scaling variables by powers of
$\varepsilon$ to eliminate mismatched labels.)

At each interpolation point, all $c^{2d}$ source products share
the same left polynomial, so we concatenate their right inputs into
a single row. The resulting packing computes $w$ targets with $w$
columns using $2w-1$ sources with $c^{2d}$ columns. For sufficiently
large $d$, $w>c^{2d}$, giving the desired reduction in the column dimension.

The left inputs are now independent, so we can rename variables
separately in each target to put all targets in the common form
$T_{a,b}$. To obtain
a recursive matrix-multiplication algorithm, we still need to reduce
the row and inner dimensions.

\paragraph{A Strassen-like identity.}
To provide the other dimension reductions, we now tensor the column reduction with two
role-exchanged copies: exchanging the inputs gives a row reduction,
and exchanging the left input with the output gives an inner-dimension
reduction. Their tensor product reduces all three matrix dimensions
simultaneously.

The resulting shape is
\[
 Q_{a,b}=T_{a,b}\otimes T_{a,b}^{12}\otimes T_{a,b}^{13}.
\]
This step reduces all three matrix dimensions, but it still increases
only the polynomial-length parameter $b$. Keeping $a$ fixed while
$b$ grows would eventually cancel the asymptotic saving. We therefore
need a way to increase $a$ as well.

For Full, exchanging the lengths $a$ and $b$ is exactly exchanging
the two input roles. Consequently, in the combined product $Q_{a,b}$, this exchange
merely reorders the factors, which amounts to renaming polynomial
variables.\footnote{For example,
$\operatorname{Mid}_{a,b}^{13}
=(\operatorname{Full}_{a,b}^{23})^{13}
=(\operatorname{Full}_{a,b}^{12})^{23}
=\operatorname{Mid}_{b,a}$,
so the Mid factor in $T_{a,b}$ pairs with the one in $T_{a,b}^{13}$.}
More generally, the gadget permits an integer shift $t\ge1$, increasing
the source length from $b$ to $b+t$, whenever $ct\ge a+2b-1$.
Thus, after increasing $b$, we can exchange the two parameters
by renaming variables and apply the construction again.
Starting with equal lengths $a=b=L$, apply one step to increase $b$
to $L'=L+t$.\footnote{For $c\ge2$, take
$t=\lceil(3L-1)/(c-1)\rceil$. Then $ct\ge3L+t-1$, so the
shift condition holds at both $(a,b)=(L,L)$ and $(L+t,L)$; moreover,
$L'\le(1+3/(c-1))L+1$.}
Rename the resulting product using this symmetry, then apply the step again:
\[
 Q_{L,L}\longrightarrow Q_{L,L'}
 \xrightarrow{\text{rename variables}} Q_{L',L}\longrightarrow Q_{L',L'}.
\]
The first and last arrows reduce all three matrix dimensions; the middle
arrow only renames variables. Together these two steps form
a \emph{round}, increasing both polynomial lengths from $L$ to $L'$.
Its sources have the same equal-length form as its targets, so rounds
can be repeated.

The resulting product is ordinary matrix multiplication at $a=b=1$,
and at $a=b=L$ can be computed using $O(L^{6d})$ matrix multiplications.
For each fixed integer $c\ge2$, choosing sufficiently
many variables and rounds gives any exponent strictly above
\[
 \frac32\,\frac{\log(1+3/(c-1))}{\log(1+2/c)}.
\]
This threshold tends to $9/4$ as $c\to\infty$, as shown in the
application of Theorem~\ref{thm:packing-identity} in the proof of
Corollary~\ref{thm:identity}.
After finitely many rounds, block recursion and interpolation
(to remove $\varepsilon$, as in~\cite{bini}) give a fixed exact
rational identity with this saving.
Clearing its denominators and applying ordinary block recursion gives
our integer-arithmetic algorithm.

\section{Polynomial packing}\label{sec:interpolation}
We here introduce our notion of polynomial packing, and show some simple ways of composing them. We start by recalling some standard notation.


\paragraph{Polynomials and their coefficients.}
We use polynomials with matrix coefficients over a fixed field of
characteristic zero; scalar coefficients are the $1\times1$ case.
An ordinary univariate polynomial has \emph{length} $a$ if
its degree is less than $a$. For every integer $k$, we write
$[z^k]A$ for the coefficient of $z^k$ in $A$, which is zero
when that monomial is absent. We represent a
polynomial of length $a$ by its $a$ coefficients at exponents
$0,\ldots,a-1$, including zero coefficients.

We also allow monomials $u=z_1^{i_1}\cdots z_m^{i_m}$ with integer
exponents $i_1,\ldots,i_m\in\mathbb Z$, and finite linear combinations
of such monomials. We write $[u]A$ for the coefficient of $u$ in $A$.
The constant monomial is $1$.

For sets of monomials $\mathcal I,\mathcal J$, write
\[
 \mathcal I\mathcal J=\{uv:u\in\mathcal I,\ v\in\mathcal J\},
 \qquad
 \mathcal I^{-1}=\{u^{-1}:u\in\mathcal I\}.
\]
We use $\prod_i\mathcal I_i$ for the iterated set product.

A permutation $\pi$ of $\{1,\ldots,m\}$ acts on monomials by
replacing each variable $z_i$ with $z_{\pi(i)}$:
\[
 \pi\!\left(z_1^{e_1}\cdots z_m^{e_m}\right)
 =z_{\pi(1)}^{e_1}\cdots z_{\pi(m)}^{e_m}.
\]
For a set of monomials, define
$\pi(\mathcal I)=\{\pi(u):u\in\mathcal I\}$.

\paragraph{Linear maps.}
A linear map on coefficient lists takes fixed scalar linear
combinations of their entries. It is represented by a matrix $E$,
and we write $E(x)=Ex$. For matrix coefficients, the coordinates
are their scalar entries. Coefficient extraction, shifts, matrix
transposition, and rearranging or concatenating matrix blocks are
linear operations.

\subsection{Selected matrix-polynomial products}\label{sec:selected-products}
Our computations multiply two matrix polynomials with prescribed
allowed monomials and retain a prescribed set of product coefficients.

\begin{definition}[Selected matrix-polynomial product]\label{def:computation}
Let $m\ge1$, put $\mathbf z=(z_1,\ldots,z_m)$, and let
$\mathcal I,\mathcal J,\mathcal K$ be finite sets of monomials
in $\mathbf z$ with integer exponents.
For positive integers $p,q,r$, the inputs are
\[
 A(\mathbf z)=\sum_{u\in\mathcal I}A_u u,
 \qquad
 B(\mathbf z)=\sum_{v\in\mathcal J}B_v v,
\]
where $A_u=[u]A$ and $B_v=[v]B$ are arbitrary matrices of
sizes $p\times q$ and $q\times r$, respectively. Coefficients outside
$\mathcal I,\mathcal J$ are zero.
The \emph{selected matrix-polynomial product}
$\operatorname{Prod}_{\mathcal I,\mathcal J,\mathcal K}[p,q,r]$
returns
\begin{equation}\label{eq:selected-product}
 \operatorname{Prod}_{\mathcal I,\mathcal J,\mathcal K}[p,q,r](A,B)
 =\sum_{w\in\mathcal K}[w](AB)\,w.
\end{equation}
Its output coefficients have size $p\times r$. We also represent
the output by its coefficient list indexed by $\mathcal K$ in a
fixed order.
We call $P=(\mathcal I,\mathcal J,\mathcal K)$ its \emph{shape}, and
write $P[p,q,r]$ for the corresponding selected product.
The monomial sets are independent of $p,q,r$, and we omit these
dimensions when they are clear from context.
\end{definition}

Ordinary matrix multiplication is the case
$\mathcal I=\mathcal J=\mathcal K=\{1\}$.
At the other extreme, taking matrix dimensions $[1,1,1]$ and
$\mathcal K=\mathcal I\mathcal J$
gives ordinary scalar polynomial multiplication.

\paragraph{Role exchanges and duality.}
We call the left input, right input, and output roles $1,2,3$, respectively.
A \emph{role exchange} swaps two of these roles. Input monomials $u,v$
contribute to an output monomial $w$ when $uv=w$. Exchanging an input
with the output inverts both monomial labels: the relation becomes
$uw^{-1}=v^{-1}$ or $w^{-1}v=u^{-1}$.
The matrix dimensions follow the shared indices: $p$ links the left
input and output, $q$ links the two inputs, and $r$ links the right
input and output. Thus the three exchanges are
\[
 \begin{aligned}
 \bigl(\operatorname{Prod}_{\mathcal I,\mathcal J,\mathcal K}[p,q,r]\bigr)^{12}
 &=\operatorname{Prod}_{\mathcal J,\mathcal I,\mathcal K}[r,q,p],\\
 \bigl(\operatorname{Prod}_{\mathcal I,\mathcal J,\mathcal K}[p,q,r]\bigr)^{23}
 &=\operatorname{Prod}_{\mathcal I,\mathcal K^{-1},\mathcal J^{-1}}[q,p,r],\\
 \bigl(\operatorname{Prod}_{\mathcal I,\mathcal J,\mathcal K}[p,q,r]\bigr)^{13}
 &=\operatorname{Prod}_{\mathcal K^{-1},\mathcal J,\mathcal I^{-1}}[p,r,q].
 \end{aligned}
\]
We call $P^{23}$ the \emph{dual} of $P$, written $\widehat P$;
it keeps the left support unchanged.
For successive exchanges, $P^\sigma$, $\sigma\in S_3$, denotes the
corresponding permutation of $(\mathcal I,\mathcal J,\mathcal K^{-1})$,
moving the support in position $i$ to position $\sigma(i)$, with
the resulting third support inverted to obtain the output support.
Matrix dimensions are permuted by the exchanges above.
Successive exchanges satisfy $(P^\sigma)^\tau=P^{\tau\sigma}$
for $\sigma,\tau\in S_3$; in particular, $\widehat{\widehat P}=P$.

\paragraph{Full and middle products.}
Throughout the polynomial constructions, write $s=a+b-1$.
For positive integers $a,b$, define the univariate full-product shape by
\[
 \operatorname{Full}_{a,b}
 :=\bigl(\{z^i:0\le i<a\},\ \{z^j:0\le j<b\},
          \ \{z^k:0\le k<a+b-1\}\bigr).
\]
Define the middle-product shape as its dual:
\begin{equation}\label{eq:full-mid-exchange}
 \operatorname{Mid}_{a,b}:=\widehat{\operatorname{Full}}_{a,b}.
\end{equation}
Attaching matrix dimensions $[p,q,r]$ gives the selected products
$\operatorname{Prod}_{\mathcal I,\mathcal J,\mathcal K}[p,q,r]$
with the supports in the following table:
\[
 \begin{array}{c|ccc}
  \text{Map} & \mathcal I\text{ (left)} & \mathcal J\text{ (right)}
             & \mathcal K\text{ (output)}\\ \hline
  \operatorname{Full}_{a,b}
   & \{z^i:0\le i<a\} & \{z^j:0\le j<b\} & \{z^k:0\le k<s\}\\
  \operatorname{Mid}_{a,b}
   & \{z^i:0\le i<a\} & \{z^j:1-s\le j\le0\} & \{z^k:1-b\le k\le0\}
 \end{array}
\]
Full returns all $s$ coefficients of $AB$; Mid returns
$[z^{1-b+j}](AQ)$ for $0\le j<b$, in that order.
This is a shifted convention for the usual middle
product~\cite[Section~4]{bls03}.

\subsection{Packings and same-left packings}
We are now ready to state our definition of polynomial packing, which formalizes how to ``pack'' a list of selected matrix-polynomial products
into another list using linear operations.

\begin{definition}[Polynomial packing]\label{def:packing}
Let $T=(T_i)$ and $S=(S_j)$ be finite lists of selected products.
An \emph{exact polynomial packing} (PP) of $T$ using $S$ consists of
rational linear maps $E^L,E^R$ from the target left- and right-input
coefficient lists to the corresponding source input lists, and a
rational linear decoding map $D$ from source output lists to target output lists.
For all allowed target inputs $A=(A_i),B=(B_i)$, they satisfy
\begin{equation}\label{eq:packing}
 D\!\left(\bigl(S_j((E^L(A))_j,(E^R(B))_j)\bigr)_j\right)
 =\bigl(T_i(A_i,B_i)\bigr)_i.
\end{equation}
The left-hand side is the output of the packing; its difference from
the target output is the error.
When all targets and sources have the same left polynomial space,
a \emph{same-left polynomial packing} (SL-PP) requires the identity
only when all $A_i=A$, and requires the left encoder to copy $A$:
\[
 \bigl(E^L(A,\ldots,A)\bigr)_j=A
 \qquad\text{for every source }j.
\]
A \emph{packing modulo $\varepsilon$} allows the three maps to have
coefficients in $\Q[\varepsilon,\varepsilon^{-1}]$ and adds an error
$\varepsilon R_\varepsilon(A,B)$ to the right-hand side
of~\eqref{eq:packing}, where $R_\varepsilon$ is polynomial in
$\varepsilon$ and the scalar input entries.
\end{definition}

\paragraph{Effort and Compression Rate.}

For a shape $P=(\mathcal I,\mathcal J,\mathcal K)$
with nonempty supports, let $a_k$ be the largest exponent of $z_k$
in $\mathcal I$ minus the smallest, plus one (i.e., the length of the
exponent range). Define $b_k$ and $h_k$
in the same way using $\mathcal J$ and $\mathcal K$.
We define a notion of ``effort" for the selected product, as follows: 
\begin{equation}\label{eq:effort}
 \begin{aligned}
 e_k&=\min\{a_k+b_k-1,\ a_k+h_k-1,\ b_k+h_k-1\},\\
 \operatorname{eff}(P)&=\prod_{k=1}^m e_k,
 \qquad
 \operatorname{eff}(P[p,q,r])=pqr\operatorname{eff}(P).
 \end{aligned}
\end{equation}
If a support is empty, set the effort to zero.
For Full and Mid, the three range lengths are $(a,b,s)$ and
$(a,s,b)$, where $s=a+b-1$. In both cases $e_1=s$, so their
effort is $spqr$.
The \emph{compression rate} of a packing is its total target effort
divided by its total source effort, i.e.
$\frac{\sum_i\operatorname{eff}(T_i)}{\sum_j\operatorname{eff}(S_j)}$,
when the denominator is positive.

The above notion of effort is motivated by the following standard result, showing that effort bounds the number of scalar multiplications needed to compute the product.
To formalize this, we consider the notion of a
\emph{rational bilinear algorithm}~\cite[Chapter~14]{bcs97}.
A map $F$ admits such an algorithm using $r$ scalar multiplications
if there exist scalar-valued rational linear maps $\ell_\nu,m_\nu$
and fixed outputs $W_\nu$ with rational coefficients such that,
for all allowed inputs $A,B$,
$F(A,B)=\sum_{\nu=1}^r\ell_\nu(A)m_\nu(B)W_\nu$.
An algorithm modulo $\varepsilon$ allows coefficients in
$\Q[\varepsilon,\varepsilon^{-1}]$ and returns
$F(A,B)+\varepsilon R_\varepsilon(A,B)$, where $R_\varepsilon$ is
polynomial in $\varepsilon$ and the scalar input entries.
For selected products, a bilinear algorithm using $r$ scalar
multiplications is precisely a PP with $r$ scalar multiplication
sources; the same correspondence holds modulo $\varepsilon$.

\begin{lemma}[Evaluation bound; cf.~{\cite[Section~4]{bls03}}]\label{lem:evaluation}
Every selected product $P$ admits an exact rational bilinear algorithm
using at most $\operatorname{eff}(P)$ scalar multiplications.
\end{lemma}

We also recall how interpolation removes the formal parameter,
following Bini~\cite{bini}. For completeness, we give a proof in
Appendix~\ref{app:block-recursion}.

\begin{lemma}[Bini's approximate-to-exact conversion~\cite{bini}]\label{lem:bini}
For each integer $s\ge1$, let $F_s(A,B)=AB$ for
$A,B\in\Q^{s\times s}$.
Let $v\ge2$, $r\ge1$ be integers and $\tau\in\mathbb R$.
If $F_v$ admits a rational bilinear algorithm modulo
$\varepsilon$ using $r<v^\tau$ scalar multiplications,
then there exists an integer $V\ge2$ such that $F_V$
admits an exact rational bilinear algorithm using fewer
than $V^\tau$ scalar multiplications.
\end{lemma}

\paragraph{Useful properties of packing.}
Polynomial packings are special cases of bilinear algorithms
\cite[Chapters~14--15]{bcs97} and inherit some of their useful
properties, which we recall here. For completeness, we provide the standard proofs in Appendix~\ref{app:operations}.

Discarding unneeded outputs and adding zero inputs let us match the
sizes of successive packings. Repeating a packing lets us process
several complete batches.

\begin{observation}[Elementary operations]\label{obs:elementary}
In a PP, we may zero-pad source matrices, delete targets, or add unused
sources. For every integer $k\ge1$, repeating a packing $k$ times
multiplies both counts by $k$ and preserves its compression rate when defined.
These operations preserve exactness or validity modulo $\varepsilon$.
Deleting targets while keeping the source list fixed also preserves SL-PPs.
\end{observation}

\paragraph{Tensoring.}
We recall the tensoring operator, specialized to our selected products.

\begin{definition}[Tensor product of selected products]\label{def:tensor}
When tensoring shapes, we first rename their variables into disjoint lists.
Let
\[
 P=\operatorname{Prod}_{\mathcal I,\mathcal J,\mathcal K}[p,q,r],
 \qquad
 P'=\operatorname{Prod}_{\mathcal I',\mathcal J',\mathcal K'}[p',q',r']
\]
be selected products in disjoint variable lists $\mathbf z$ and
$\mathbf z'$.
Define the tensor product by
\[
 P\otimes P'
 :=\operatorname{Prod}_{\mathcal I\mathcal I',\,
    \mathcal J\mathcal J',\,\mathcal K\mathcal K'}
       [pp',qq',rr'].
\]
Disjoint variable lists ensure that each product monomial uniquely
determines its two factors. Tensoring with ordinary matrix multiplication
multiplies the matrix dimensions and leaves the polynomial supports unchanged.
\end{definition}

\begin{lemma}[Tensoring packings {\cite[Chapters~14--15]{bcs97}}]\label{lem:combine}
Suppose PPs pack $(T_i)_{i=1}^m$ using $(S_j)_{j=1}^n$, and
$(T'_k)_{k=1}^{m'}$ using $(S'_\ell)_{\ell=1}^{n'}$.
Then there exists a PP with $mm'$ targets and $nn'$ sources,
\[
 (T_i\otimes T'_k)_{i,k}
 \quad\text{using}\quad
 (S_j\otimes S'_\ell)_{j,\ell}.
\]
Exactness or validity modulo $\varepsilon$ is preserved.
If the outputs of the two packings have degrees at most
$d_\varepsilon,d'_\varepsilon$ in $\varepsilon$, respectively,
the output of the combined packing has degree
at most $d_\varepsilon+d'_\varepsilon$.
The same holds for SL-PPs.
For all selected products $P,P'$,
\begin{equation}\label{eq:effort-tensor}
 \operatorname{eff}(P\otimes P')
 =\operatorname{eff}(P)\operatorname{eff}(P').
\end{equation}
Consequently, tensoring multiplies compression rates.
\end{lemma}

\paragraph{Role exchanges.}
We recall the standard transposition principle for bilinear
algorithms~\cite[Chapter~13]{bcs97}.
\begin{lemma}[Role exchanges for packings]\label{lem:role-exchanges}
A PP of $(T_i)_{i=1}^m$ using $(S_j)_{j=1}^n$ gives, for every
permutation $\sigma\in S_3$ of the three roles, a PP of
\[
 (T_i^\sigma)_{i=1}^m
 \quad\text{using}\quad (S_j^\sigma)_{j=1}^n.
\]
The counts, the effort of each product, and exactness or validity
modulo $\varepsilon$ are preserved.
For $\sigma=(23)$, the same conclusion holds for SL-PPs.
\end{lemma}

\paragraph{Composition.}

We can replace each source computation by another packing.\footnote{We use
$(f\circ g)(x)=f(g(x))$: apply $g$ first, then $f$.}

\begin{lemma}[Composition {\cite[Chapters~14--15]{bcs97}}]\label{lem:composition}
PPs of $(T_i)_{i=1}^m$ using $(S_j)_{j=1}^n$ and of $(S_j)$ using
$(R_k)_{k=1}^{\ell}$ give a PP of $(T_i)$ using $(R_k)$,
with $m$ targets and $\ell$ sources.
Exactness or validity modulo $\varepsilon$ is preserved.
If all three lists have positive total effort, the compression rate
of the composed packing is the product of the two compression rates.
\end{lemma}

\paragraph{Isomorphism.}
Two shapes are \emph{isomorphic}, written $P\cong Q$, if they
differ only by renaming polynomial variables. Formally,
$(\mathcal I,\mathcal J,\mathcal K)\cong
(\mathcal I',\mathcal J',\mathcal K')$ if a permutation $\pi$
of the variables satisfies
\[
 (\mathcal I',\mathcal J',\mathcal K')=
 (\pi(\mathcal I),\pi(\mathcal J),\pi(\mathcal K)).
\]
We use the same notation for selected products with the same matrix
dimensions and isomorphic shapes.

\begin{lemma}[Isomorphism]\label{lem:isomorphism}
Isomorphism is an equivalence relation. For all shapes $P,Q,R,P',Q'$
and $\sigma\in S_3$,
\[
 \begin{aligned}
 (P\otimes Q)\otimes R&\cong P\otimes(Q\otimes R),&
 P\otimes Q&\cong Q\otimes P,\\
 (P\otimes Q)^\sigma&\cong P^\sigma\otimes Q^\sigma.
 \end{aligned}
\]
If $P\cong Q$, then $P^\sigma\cong Q^\sigma$; if also $P'\cong Q'$, then
$P\otimes P'\cong Q\otimes Q'$.
Replacing any target or source of a PP by an isomorphic selected
product preserves the counts and exactness or validity modulo
$\varepsilon$. Applying one common variable renaming to all tasks
also preserves SL-PPs.
\end{lemma}

\section{Constructing polynomial packings}\label{sec:constructions}

We construct the primal and dual packings, combine them on separate
variables, and then remove the same-left requirement.

\subsection{The primal and dual gadgets}
Our first packing extends the three-for-one gadget
of~\cite[Lemma 4.2]{original}. For $ct\ge a+2b-1$, it computes
$c+1$ full products and one middle product from $c$ full products
with $t$ additional coefficients in each source right input.

\begin{theorem}[Primal packing]\label{lem:primal}
For all integers $a,b,c,t,p,q,r\ge1$ with $ct\ge a+2b-1$, there exists
an SL-PP modulo $\varepsilon$ of $c+1$ copies of
$\operatorname{Full}_{a,b}[p,q,r]$ and one middle product
$\operatorname{Mid}_{a,b}[p,q,r]$ using $c$ copies of
$\operatorname{Full}_{a,b+t}[p,q,r]$.
\end{theorem}
\begin{proof}
We are given $A$ of length $a$, polynomials $P_1,\ldots,P_{c+1}$
of length $b$, and $Q$ supported on $\{z^j:1-s\le j\le0\}$.
Set $\widetilde Q=z^{s-1}Q$, which has length $s$.
We will encode only the right
inputs, multiply the encoded polynomials by the unchanged $A$, and
recover the requested coefficients from these products.

\paragraph{Encoding procedure.}
Since $ct\ge s+b>s$, we can split the coefficient list of $\widetilde Q$
into $c$ consecutive blocks of length $t$, padding with zeros:
\[
 Q^{(j)}(z)=\sum_{v=0}^{t-1}[z^{t(j-1)+v}]\widetilde Q\,z^v,
 \qquad
 \widetilde Q(z)=\sum_{j=1}^c z^{t(j-1)}Q^{(j)}(z).
\]
For $1\le j\le c$, set
\[
 E_j=P_j-z^tP_{j+1}+\varepsilon z^bQ^{(j)}.
\]
The first two terms combine adjacent right inputs using a shift
of $t$. The last term places one block of $\varepsilon\widetilde Q$ in
positions $b,\ldots,b+t-1$. All three terms have degree less than
$b+t$, so each $E_j$ is a valid right input to
$\operatorname{Full}_{a,b+t}$.

Compute the $c$ full source products $W_j=AE_j$, each of length
$s+t$.

\paragraph{Decoding procedure.}
From these products, form the following telescoping sum:
\begin{equation}\label{eq:telescope}
 \begin{aligned}
 W&=\sum_{j=1}^c z^{t(j-1)}W_j\\
  &=\sum_{j=1}^c z^{t(j-1)}AP_j
    -\sum_{j=1}^c z^{tj}AP_{j+1}
    +\varepsilon z^b A\sum_{j=1}^c z^{t(j-1)}Q^{(j)}\\
  &=AP_1+\varepsilon z^bA\widetilde Q-z^{ct}AP_{c+1}.
 \end{aligned}
\end{equation}
\paragraph{Recover the middle product.}
The first endpoint product $AP_1$ has no coefficients at exponents
$s$ or above. The last endpoint product $z^{ct}AP_{c+1}$ has no
coefficients below $ct\ge s+b$. Thus neither contributes at
exponents $s,\ldots,s+b-1$. For $0\le k<b$, we obtain
\[
 \varepsilon^{-1}[z^{s+k}]W
   =[z^{a-1+k}](A\widetilde Q)
   =[z^{1-b+k}](AQ).
\]
These are exactly the $b$ coefficients of
$\operatorname{Mid}_{a,b}(A,Q)$.

\paragraph{Recover the full products.}
First take the first $s$ coefficients of $W$:
\[
 Z_1=\sum_{k=0}^{s-1}[z^k]W\,z^k.
\]
Equation~\eqref{eq:telescope} gives
$Z_1\equiv AP_1\pmod{\varepsilon}$: the final endpoint product
does not reach this range, and the contribution from $\widetilde Q$ is
divisible by $\varepsilon$.

Once $Z_j$ is known, subtract $W_j$, discard the first $t$
coefficients, and take the next $s$ coefficients. In formulas,
define the length-$s$ polynomial $Z_{j+1}$ by
\[
 [z^k]Z_{j+1}=[z^{k+t}](Z_j-W_j),
 \qquad 0\le k<s,\quad 1\le j\le c.
\]
To verify this step, suppose $Z_j\equiv AP_j\pmod{\varepsilon}$.
Then
\[
 \begin{aligned}
 Z_j-W_j
   &=Z_j-AP_j+z^tAP_{j+1}-\varepsilon z^bAQ^{(j)}\\
   &\equiv z^tAP_{j+1}\pmod{\varepsilon}.
 \end{aligned}
\]
Thus $Z_{j+1}\equiv AP_{j+1}\pmod{\varepsilon}$, and induction gives
$Z_j\equiv AP_j\pmod{\varepsilon}$ for all $1\le j\le c+1$.

The encoder and decoder are linear with coefficients in
$\Q[\varepsilon,\varepsilon^{-1}]$.
The $Z_j$ are polynomial in $\varepsilon$, and the middle output
is exact, so these formulas give the required packing modulo $\varepsilon$.
\end{proof}

For $c=1$ and $t=a+2b-1$, this specializes to the three-for-one
degeneration of~\cite[Lemma~4.2]{original}.

\subsection{Amplifying compression through multiple variables}

We next pair a packing with its dual using Lemma~\ref{lem:role-exchanges}
and tensor the result to amplify the compression rate. For $d\ge1$, let
\begin{equation}\label{eq:paired-shape}
 T^d_{a,b}:=(\operatorname{Full}_{a,b}\otimes
                 \operatorname{Mid}_{a,b})^{\otimes d}.
\end{equation}
When $d$ is clear from context, we write $T_{a,b}=T^d_{a,b}$.

Tensoring a packing with its dual squares its compression rate.
Taking $d$ tensor copies of the result raises the original rate to
the power $2d$. Retaining only targets with equally many Full and Mid
factors loses at most a factor $d+1$, so an initial compression rate
greater than one can be amplified exponentially.

\begin{lemma}[Same-left amplification]\label{lem:sl}
Let $d,m,n,a,b,b',p,q,r\ge1$ be integers. Suppose an SL-PP packs
$m$ targets, each $\operatorname{Full}_{a,b}[p,q,r]$ or
$\operatorname{Mid}_{a,b}[p,q,r]$, using $n$ copies of
$\operatorname{Full}_{a,b'}[p,q,r]$. Let
\begin{equation}\label{eq:counts}
 u=n^{2d},\qquad
 w=\left\lceil\frac{m^{2d}}{d+1}\right\rceil.
\end{equation}
Then there exist shapes $U_i\cong T_{a,b}$ and an SL-PP of
\[
 \bigl(U_i[(pq)^d,(pq)^d,r^{2d}]\bigr)_{i=1}^w
\]
using $u$ copies of
\[
 T_{a,b'}[(pq)^d,(pq)^d,r^{2d}].
\]
Exactness or validity modulo $\varepsilon$ is preserved.
\end{lemma}

We first show a simple counting fact providing a lower bound on how many products we can retain.

\begin{claim}[Counting words by type]\label{lem:count-words}
Let $f,g\ge0$ and $d\ge1$ be integers, with $m=f+g\ge1$, and put
\[
 \begin{aligned}
 \mathcal P&=\{\mathsf F_1,\ldots,\mathsf F_f,
               \mathsf M_1,\ldots,\mathsf M_g\},\\
 \mathcal D&=\{\mathsf F_1,\ldots,\mathsf F_g,
               \mathsf M_1,\ldots,\mathsf M_f\}.
 \end{aligned}
\]
Every $\mathsf F$-letter has type Full and every $\mathsf M$-letter
has type Mid. Let $M_d$ be the number of words
whose first $d$ letters lie in $\mathcal P$, whose last $d$ letters
lie in $\mathcal D$, and which contain exactly $d$ letters of each
type. Then
\[
 M_d\ge\frac{m^{2d}}{d+1}.
\]
\end{claim}
\begin{proof}
For $0\le h\le d$, a first half with exactly $h$ Mid letters
can be chosen in
\[
 N_h=\binom dh g^h f^{d-h}
\]
ways: choose the $h$ Mid positions, fill each with one of the $g$
Mid letters in $\mathcal P$, and fill each remaining position with
one of its $f$ Full letters.

To obtain $d$ letters of each type overall, the second half must
have exactly $h$ Full letters. Since $\mathcal D$ has $g$ Full
letters and $f$ Mid letters, it has the same $N_h$ choices.
Each of these can follow any of the $N_h$ first halves. Thus
\[
 M_d=\sum_{h=0}^d N_h^2
 \ge\frac{\left(\sum_{h=0}^d N_h\right)^2}{d+1}
 =\frac{(f+g)^{2d}}{d+1}
 =\frac{m^{2d}}{d+1}
\]
where the inequality follows from Cauchy--Schwarz, the next equality
from the binomial theorem, and the final equality from $m=f+g$.
\end{proof}

\begin{proof}[Proof of Lemma~\ref{lem:sl}]
Apply Lemma~\ref{lem:role-exchanges} with $\sigma=(23)$ to obtain the dual
packing. Tensor the original packing with its dual, then tensor $d$ copies
of the result.
The dual has matrix dimensions $[q,p,r]$, so pairing gives
$[pq,pq,r^2]$ and the $d$ tensor copies give
$[(pq)^d,(pq)^d,r^{2d}]$.
By Lemma~\ref{lem:combine}, the result is an SL-PP with $m^{2d}$
targets and $n^{2d}$ sources of shape $T_{a,b'}$ at these dimensions,
preserving exactness or validity modulo $\varepsilon$.
Its targets are indexed by pairs of length-$d$ words from the
original target list and its dual. Since Full and Mid are each other's
duals, $f$ Full and $g$ Mid targets in the original
list give $g$ Full and $f$ Mid targets in its dual, with $f+g=m$.
After labeling these targets as in Claim~\ref{lem:count-words}, the claim
supplies at least $w$ targets with exactly $d$ factors of each type, hence
shapes $U_i\cong T_{a,b}$ by Lemma~\ref{lem:isomorphism}.
Retain $w$ of them by Observation~\ref{obs:elementary}.
\end{proof}

The target shapes $U_i$ may differ. Putting each target in
canonical form may therefore require different variable renamings,
which also change its left input. We first make the left inputs
independent so that these changes can be chosen separately.

\subsection{Removing the same-left requirement}

We now show how to turn a same-left polynomial packing (SL-PP) into
a general polynomial packing (PP), allowing independent left inputs.
OpenAI~\cite[Section~3]{original} proved an analogous result using
finite Fourier projection followed by scaling variables
by powers of $\varepsilon$ to eliminate mismatched labels.
We give a direct and elementary proof using polynomial interpolation.

\begin{lemma}[From SL-PP to PP]\label{lem:sl-to-pp}
Let $u,w,p,q,r\ge1$ be integers and $F_1,\ldots,F_w,S$ be shapes.
An SL-PP modulo $\varepsilon$ of
\[
 \bigl(F_i[p,q,r]\bigr)_{i=1}^w
 \quad\text{using }u\text{ copies of }S[p,q,r]
\]
gives a PP modulo $\varepsilon$ of the
targets $\bigl(F_i[p,q,wr]\bigr)_{i=1}^w$ using $2w-1$ copies of
$S[p,q,ur]$.
\end{lemma}

The transformation multiplies the compression rate by $w/(2w-1)$, so
allowing independent left inputs costs less than a factor of two.

\begin{proof}
Since every source has the same left input $A$,
\[
 S[p,q,ur](A,[\,V_1\ \cdots\ V_u\,])
 =[\,S[p,q,r](A,V_1)\ \cdots\ S[p,q,r](A,V_u)\,].
\]
Concatenating the $u$ encoded right inputs gives $q\times ur$
coefficient matrices. Splitting the output into column blocks
before applying the original decoder gives an SL-PP using one
copy of $S[p,q,ur]$.

Write $B_i=[\,B_{i,0}\ \cdots\ B_{i,w-1}\,]$, with $q\times r$
coefficient matrices in each block. It suffices to recover the
block products $F_i[p,q,r](A_i,B_{i,j})$ and concatenate them.

Interpolate coefficientwise polynomials $A(\xi),B_i'(\xi)$ of degree
at most $w-1$ satisfying\footnote{Lagrange interpolation at fixed
points is linear in the prescribed values, as are evaluation at a
fixed point and coefficient extraction from the resulting polynomial.}
\[
 A(i)=A_i,\qquad B_i'(i+\varepsilon j)=B_{i,j}
 \quad(1\le i\le w,\ 0\le j<w).
\]
These interpolations preserve the allowed monomials.\footnote{The encoding at the points
$i+\varepsilon j$ was proposed by Codex after I
prompted it to prove this result using Reed--Muller
decoding~\cite[Section~2.2]{sudan02}.}
The coefficients of $A$ are independent of $\varepsilon$; the
Lagrange formula for $B_i'$ has denominators equal to nonzero
rational multiples of $\varepsilon^{w-1}$. Thus $\varepsilon^{w-1}B_i'$
has coefficients polynomial in $\varepsilon$, and
$A(i+\varepsilon j)\equiv A_i\pmod{\varepsilon}$.

Let $E^R_\eta,D_\eta$ be the right encoder and decoder of
this one-source packing, with its parameter denoted by $\eta$.
Since the left encoder is the identity,
at each $k=0,\ldots,2w-2$ compute
\[
 D_{\varepsilon^w}\!\left(
 S[p,q,ur]\!\left(A(k),E^R_{\varepsilon^w}\bigl((B_i'(k))_{i=1}^w\bigr)\right)
 \right).
\]
This uses one copy of $S[p,q,ur]$ per point. By
Definition~\ref{def:computation}, each source is linear in its right
input. Since the encoders and decoder are linear, the packing output
and its error are also linear in the right inputs.
Before substituting the right inputs, the error
is divisible by $\varepsilon^w$. Since it is linear in those inputs,
their denominators cancel at most $w-1$ powers of $\varepsilon$.
Hence these values
agree modulo $\varepsilon$ with
$P_i(k)=F_i(A(k),B_i'(k))$.
Since $F_i$ multiplies inputs of degree at most $w-1$ in $\xi$
and selects coefficients, $\deg_\xi P_i\le2w-2$.
Interpolation from these $2w-1$ values therefore gives
\[
 \widetilde P_i(\xi)=P_i(\xi)+\varepsilon R_i(\xi,\varepsilon),
\]
where $R_i$ is polynomial in $\xi,\varepsilon$.
Evaluating at $\xi=i+\varepsilon j$ yields
\[
 \widetilde P_i(i+\varepsilon j)
 \equiv F_i(A(i+\varepsilon j),B_{i,j})
 \equiv F_i(A_i,B_{i,j})\pmod{\varepsilon}.
\]
All encoders and decoders are linear with coefficients in
$\Q[\varepsilon,\varepsilon^{-1}]$. Concatenating the recovered blocks gives the
$w$ independent targets with the original polynomial shapes.
\end{proof}

\section{A fixed matrix multiplication identity}\label{sec:conversion}

We iterate the round constructed below until the coefficient matrices
have fixed size, then evaluate the remaining polynomial products
by interpolation.

\subsection{From columns to all matrix dimensions}

We first define the shape on which the round will operate and
establish its symmetry in $a,b$:
$Q_{a,b}\cong Q_{b,a}$.
This will let us use the same reduction to increase either length.
The key identity is
\begin{equation}\label{eq:full-input-exchange}
 \operatorname{Full}_{b,a}=\operatorname{Full}_{a,b}^{12}.
\end{equation}
The combined shape contains every role permutation of Full
equally often, so exchanging $a$ and $b$ merely reorders its
tensor factors. 

\begin{lemma}[The product for a round]\label{lem:round-shape}
For each integer $d\ge1$, let $T_{a,b}$ be as in~\eqref{eq:paired-shape}
and put
\[
 Q_{a,b}:=T_{a,b}\otimes T_{a,b}^{12}\otimes T_{a,b}^{13}.
\]
For all integers $a,b\ge1$,
\[
 Q_{a,b}\cong Q_{b,a},\qquad
 \operatorname{eff}(Q_{a,b})=(a+b-1)^{6d}.
\]
For every $p,q,r\ge1$, $Q_{1,1}[p,q,r]$ is ordinary matrix
multiplication with dimensions $[p,q,r]$.
\end{lemma}
\begin{proof}
Equation~\eqref{eq:paired-shape} and Lemma~\ref{lem:isomorphism} give
\[
 \begin{aligned}
 Q_{a,b}\cong\bigl[
 &(\operatorname{Full}_{a,b}\otimes\operatorname{Mid}_{a,b})\\
 &\otimes(\operatorname{Full}_{a,b}^{12}\otimes\operatorname{Mid}_{a,b}^{12})\\
 &\otimes(\operatorname{Full}_{a,b}^{13}\otimes\operatorname{Mid}_{a,b}^{13})
 \bigr]^{\otimes d}.
 \end{aligned}
\]
By~\eqref{eq:full-input-exchange},
$\operatorname{Full}_{a,b}^{12}=\operatorname{Full}_{b,a}$.
Using~\eqref{eq:full-mid-exchange} and the convention for successive
role exchanges, the two exchanged Mid factors satisfy
\[
 \begin{aligned}
 \operatorname{Mid}_{a,b}^{12}
 &=(\operatorname{Full}_{a,b}^{23})^{12}
  =(\operatorname{Full}_{a,b}^{12})^{13}
  =\operatorname{Full}_{b,a}^{13},\\
 \operatorname{Mid}_{a,b}^{13}
 &=(\operatorname{Full}_{a,b}^{23})^{13}
  =(\operatorname{Full}_{a,b}^{12})^{23}
  =\operatorname{Mid}_{b,a}.
 \end{aligned}
\]
Substituting and regrouping by Lemma~\ref{lem:isomorphism} gives
\[
 \begin{aligned}
 Q_{a,b}\cong\bigl[
 &(\operatorname{Full}_{a,b}\otimes\operatorname{Full}_{b,a})\\
 &\otimes(\operatorname{Mid}_{a,b}\otimes\operatorname{Mid}_{b,a})\\
 &\otimes(\operatorname{Full}_{a,b}^{13}\otimes\operatorname{Full}_{b,a}^{13})
 \bigr]^{\otimes d}.
 \end{aligned}
\]
Exchanging $a,b$ reverses each pair. By Lemma~\ref{lem:isomorphism},
this only renames variables, proving $Q_{a,b}\cong Q_{b,a}$.
Each Full factor has scalar effort $a+b-1$; Lemma~\ref{lem:role-exchanges}
and~\eqref{eq:effort-tensor} therefore give
$\operatorname{eff}(Q_{a,b})=(a+b-1)^{6d}$.
At $a=b=1$, every Full and Mid support is $\{1\}$, which role
exchanges and tensoring preserve. Thus $Q_{1,1}[p,q,r]$ is
ordinary matrix multiplication by Definition~\ref{def:computation}.
\end{proof}

We now construct a packing for $Q_{a,b}$ from the packing for
$T_{a,b}$. This packing reduces all three matrix dimensions
simultaneously. Applying it twice, using
$Q_{a,b}\cong Q_{b,a}$ between applications, gives a round
that increases both polynomial lengths.

\begin{lemma}[A round]\label{lem:round}
Fix an integer $d\ge1$ and use the shapes $T_{a,b},Q_{a,b}$ above.
Let $L,L',u,w\ge1$ be integers with $w\ge u$.
Suppose that, for both
$(a,b,b')=(L,L,L')$ and $(L',L,L')$, a PP, exact or modulo
$\varepsilon$, packs $w$ copies of
$T_{a,b}[1,1,w]$ using at most $2w$ copies of $T_{a,b'}[1,1,u]$.
Put $\Lambda=\lfloor w/u\rfloor^2$.
Then, for every integer $x\ge1$, there is a PP of $w^3$ copies of
\[
 Q_{L,L}[w^2\Lambda x,w^2\Lambda x,w^2\Lambda x]
\]
using $64w^3$ copies of $Q_{L',L'}[w^2x,w^2x,w^2x]$.
Exactness or validity modulo $\varepsilon$ is preserved.
\end{lemma}
\begin{proof}
For either triple $(a,b,b')$ and any integers $p,q,r\ge1$, tensor
the given packing with its $(12)$ and $(13)$ transposes and the
identity packing for $\operatorname{Full}_{1,1}[p,q,r]$.
The last factor changes only matrix dimensions by
Definition~\ref{def:tensor}. Lemmas~\ref{lem:role-exchanges}
and~\ref{lem:combine} give $w^3$ copies of
$Q_{a,b}[wp,wq,wr]$ using at most $(2w)^3=8w^3$ copies of
$Q_{a,b'}[up,uq,ur]$. Pad the source count to $8w^3$ by
Observation~\ref{obs:elementary}.

Apply the first packing at $p=q=r=w\Lambda x$. Its $8w^3$
sources have shape $Q_{L,L'}$ and square coefficient matrices
of side $wu\Lambda x$. By Lemmas~\ref{lem:round-shape}
and~\ref{lem:isomorphism}, we may rename these sources to
shape $Q_{L',L}$.

The second packing, at $p=q=r=u\Lambda x$, handles $w^3$
such products using $8w^3$ sources of shape $Q_{L',L'}$
and side $u^2\Lambda x$. Repeat it eight times
(Observation~\ref{obs:elementary}) to handle all $8w^3$
intermediate products, using $64w^3$ sources in total.
Composing with the first packing
(Lemma~\ref{lem:composition}) gives the required $w^3$
targets of shape $Q_{L,L}$ and side $w^2\Lambda x$.
Pad the final sources to side $w^2x$, since
$u^2\Lambda\le w^2$.
All operations preserve exactness or validity modulo $\varepsilon$.
\end{proof}

\subsection{Construction of a Strassen-type identity}

We finally use our PP to provide a Strassen-like matrix multiplication identity.

\begin{theorem}[A finite identity from scalar polynomial packings]\label{thm:packing-identity}
Fix integers $m>n\ge2$.
Suppose that, for every triple of integers $a,b,\delta\ge1$ with
$n\delta\ge a+2b-1$, there exists an SL-PP modulo
$\varepsilon$ of $m$ scalar targets, each of shape
$\operatorname{Full}_{a,b}$ or $\operatorname{Mid}_{a,b}$,
using $n$ scalar sources of shape $\operatorname{Full}_{a,b+\delta}$.
Then, for every
\[
 \tau>\frac{3\log(1+3/(n-1))}{2\log(m/n)},
\]
there exists an exact rational bilinear algorithm
for matrix multiplication of some side $v\ge2$ using $r<v^\tau$
scalar multiplications.
\end{theorem}

\begin{proof}
For every $L\ge1$, put
\[
 \delta_L=\left\lceil\frac{3L-1}{n-1}\right\rceil,
 \qquad L'=L+\delta_L,
 \qquad \rho=1+\frac3{n-1}.
\]
Then $n\delta_L\ge3L+\delta_L-1$, so the hypothesis supplies
the two SL-PPs with $(a,b)=(L,L)$ and $(L',L)$, both with
source right-input length $L'$. Also $L'\le\rho L+1$.

For each $d\ge1$, take $u,w$ from~\eqref{eq:counts} and put
$\Lambda=\lfloor w/u\rfloor^2$. As $d\to\infty$, these formulas
give $w/u\to\infty$ and $\log\Lambda=4d\log(m/n)+O(\log d)$.
Consequently,
\[
 \frac{\log64+6d\log\rho}{\log\Lambda}
 \longrightarrow\frac{3\log\rho}{2\log(m/n)}<\tau.
\]
Fix $d$ large enough that
\begin{equation}\label{eq:packing-growth}
 64\rho^{6d}<\Lambda^\tau.
\end{equation}
Since $\tau>0$, this also gives $\Lambda>1$ and hence $w\ge u$.

For each $a\in\{L,L'\}$, apply Lemmas~\ref{lem:sl}
and~\ref{lem:sl-to-pp} at $p=q=r=1$, then rename the targets by
Lemma~\ref{lem:isomorphism}. This gives $w$ copies of
$T_{a,L}[1,1,w]$ using $2w-1$ copies of $T_{a,L'}[1,1,u]$.
Since $2w-1\le2w$, the hypotheses of Lemma~\ref{lem:round} hold,
with counts independent of $L,L'$.
By Lemma~\ref{lem:round-shape} and~\eqref{eq:effort}, each round
has compression rate
\[
 \frac{\Lambda^3}{64}
  \left(\frac{2L-1}{2L'-1}\right)^{6d},
\]
independent of $x$.

Starting from $L_0=1$, set $L_{j+1}=L_j+\delta_{L_j}$.
Since $L_{j+1}\le\rho L_j+1$ and $\rho>1$, we have
$L_k=O(\rho^k)$.
Fix $k\ge1$ and put $v_k=w^2\Lambda^k$.
Apply the rounds $(L_j,L_{j+1})$ with $x=\Lambda^{k-j-1}$
for $j=0,\ldots,k-1$. Successive shapes and matrix dimensions
match. The fixed counts $w^3,64w^3$ give $64$ complete batches for
the next round from each batch in the current round.
Repetition (Observation~\ref{obs:elementary}) and composition
(Lemma~\ref{lem:composition}) therefore pack $w^3$ copies of
$Q_{1,1}[v_k,v_k,v_k]$ using sources
$Q_{L_k,L_k}[w^2,w^2,w^2]$, with compression rate
\[
 C_k=\prod_{j=0}^{k-1}
 \left[\frac{\Lambda^3}{64}
 \left(\frac{2L_j-1}{2L_{j+1}-1}\right)^{6d}\right]
 =\frac{\Lambda^{3k}}{64^k(2L_k-1)^{6d}}.
\]
The length factors telescope, and validity modulo $\varepsilon$
is preserved throughout.

By Lemma~\ref{lem:round-shape}, the targets are matrix products,
with total effort $w^3v_k^3$. The definition of compression rate
and Lemma~\ref{lem:evaluation} therefore bound the scalar
multiplications needed to evaluate all sources by
\[
 E_k:=\frac{w^3v_k^3}{C_k}
     =64^k w^9(2L_k-1)^{6d}
     \le C\bigl(64\rho^{6d}\bigr)^k,
\]
where $C$ depends only on the fixed data and $d$.
Compose the evaluations by
Lemma~\ref{lem:composition} and retain one target by
Observation~\ref{obs:elementary}.
This computes one matrix product of side $v_k$; retaining it leaves
the constant batch factor $w^3$ in the cost.
By~\eqref{eq:packing-growth}, a sufficiently large computable $k$
gives $E_k<v_k^\tau$. Lemma~\ref{lem:bini} gives the claimed exact
identity.
\end{proof}

For the primal packing, the two required length increases follow
from the same inequality. The theorem then gives the desired exponent.

\begin{corollary}[A fixed matrix multiplication identity]\label{thm:identity}
For every $\tau>9/4$, there exist computable integers $v\ge2$, $r\ge1$
with $r<v^\tau$ and an exact rational bilinear algorithm for
matrix multiplication of side $v$ using $r$ scalar multiplications.
\end{corollary}
\begin{proof}
For every integer $c\ge2$, Theorem~\ref{lem:primal} at $p=q=r=1$
and $t=\delta$ supplies the packings required by
Theorem~\ref{thm:packing-identity} with $m=c+2$ and $n=c$.
The exponent threshold is
\[
 \frac32\,\frac{\log(1+3/(c-1))}{\log(1+2/c)}
 \longrightarrow\frac94.
\]
For any prescribed $\tau>9/4$, choose $c$ large enough that this
threshold is below $\tau$, and apply the theorem.
\end{proof}

\subsection{Integer arithmetic}

We finish by applying standard block recursion to the fixed identity,
as in Strassen's algorithm~\cite{strassen69}. Its parameters and
coefficients are constants once the requested exponent is fixed.

\begin{proof}[Proof of Theorem~\ref{thm:main}]
Choose the identity in Corollary~\ref{thm:identity} with
$\tau=9/4+\gamma$. Clear the fixed denominators in the encoders and
decoder to obtain integer linear maps whose output is
$\Delta AB$ for a fixed positive integer $\Delta$.

For $n=v^j$, partition each input into a $v\times v$ array of
blocks of side $n/v$ and apply the fixed identity to these blocks.
This substitution is valid because, after expanding the linear
combinations in the bilinear algorithm, every term contains one entry
of $A$ followed by one entry of $B$. Equality of their scalar
coefficients therefore also gives the identity for matrix blocks.
Compute its $r$ block products recursively, using the integer
encodings and decodings and dividing each output entry exactly
by $\Delta$ at every level, until the blocks are scalars.
All operations outside the recursive products cost $O(n^2)$, so
\[
 T(n)\le rT(n/v)+O(n^2)=O(n^\tau),
\]
since $r<v^\tau$ and $\tau>2$. Padding to the next power of $v$
handles every positive integer $n$.
\end{proof}

For $B$-bit integer inputs, encoding adds only $O(1)$ bits per
level. Products have at most twice the resulting input length,
and decoding adds $O(\log n)$ bits over the remaining levels.
All intermediate integers therefore have $O(B+\log n)$ bits.
The bit cost is
$O(n^{9/4+\gamma}\mathsf M(B+\log n))$, where $\mathsf M(b)$
is the cost of multiplying $b$-bit integers; multiplication and
exact division by the fixed constants take linear time.

\begin{remark}[Other products and their duals]
Nothing after the primal construction requires Full and Mid specifically.
We may replace them by shapes $P_{a,b}$ and their duals, provided that
$P_{b,a}\cong P_{a,b}^{12}$ for all integers $a,b\ge1$ and
$P_{1,1}[p,q,r]$ is ordinary matrix multiplication for all integers
$p,q,r\ge1$. These two properties give the parameter symmetry and
matrix-multiplication base case needed for recursion.

Fix integers $m>n\ge1$ and constants $\rho>1$, $C\ge0$.
Suppose $\operatorname{eff}(P_{L,L})=O(L^\kappa)$ for some fixed
$\kappa>0$.
Suppose that for every integer $L\ge1$ there exists an integer
$1\le L'\le\rho L+C$ such that, for both $a=L$ and $a=L'$,
there exists an SL-PP modulo $\varepsilon$ of $m$ scalar targets,
each $P_{a,L}$ or $\widehat P_{a,L}$, using $n$ scalar sources
$P_{a,L'}$.

The proofs of Lemmas~\ref{lem:sl}, \ref{lem:sl-to-pp},
\ref{lem:round-shape}, and~\ref{lem:round} apply with these replacements.
Since $\operatorname{eff}(Q_{L,L})=\operatorname{eff}(P_{L,L})^{6d}$,
the same iteration as in Theorem~\ref{thm:packing-identity} gives,
for every
\[
 \tau>\frac{3\kappa\log\rho}{2\log(m/n)},
\]
there exists an exact rational bilinear algorithm for matrix
multiplication using $r<v^\tau$ scalar multiplications for some $v\ge2$.
In particular, exponents arbitrarily close to $2$ would follow
from such packings with
$\kappa\log\rho/\log(m/n)\longrightarrow4/3$.
\end{remark}

\paragraph{Lean formalization.}
In a single run,
\href{https://aristotle.harmonic.fun/projects/49c70ea6-754e-4480-aefa-df75bf515f07}{Aristotle}
generated Lean~4 code to formalize the rational
matrix-multiplication identities and the resulting arithmetic-operation
bound~\cite{lean-formalization}, using Mathlib's standard definitions of matrices and matrix
multiplication.
(We have not independently verified the generated code or
its correspondence with the paper.)
The formalization comprises approximately 3,600
lines, compared with approximately 20,700 lines in OpenAI's
development (which also includes substantial general-purpose infrastructure) for the $9/4$ bound.
It proves that, for every $\gamma>0$, every positive size $n$ admits
a correct program for rational matrix multiplication using
$O(n^{9/4+\gamma})$ scalar operations, including all operations in
the linear encodings and decoding. Beyond its compactness, the
formalization certifies the rational identities underlying our
construction: all coefficients and program constants lie
in $\Q$. OpenAI's formalized arithmetic bound is over $\mathbb C$.
Rational coefficients allow our identities to be implemented exactly
by clearing denominators, without approximating complex constants.

\paragraph{AI Use Statement.}
OpenAI Codex (Astra) was used extensively to study~\cite{original},
learn about the problem, and explore candidate constructions.
While the main ideas behind the constructions and abstractions
are my own, Codex was invaluable in exploring definitions and proof
approaches and in formalizing proof ideas and sketches. Most of
the text has been written and revised with Codex through detailed and repeated
prompts from me.
Claude provided feedback on the manuscript and its presentation.

\appendix
\section{Standard results}\label{app:standard}

\subsection{Supporting lemmas}\label{app:supporting}

\begin{lemma}[Linearity in the inputs]\label{lem:input-linearity}
The output of a PP and its error are linear in either input list
when the other is fixed. For an SL-PP, the same holds for the
common left input and the right-input list.
\end{lemma}
\begin{proof}
Selected products are linear in either input when the other is
fixed by~\eqref{eq:selected-product}.
Applying separate linear encoders and a linear decoder, identifying
the left inputs, and subtracting the target output preserve this property.
\end{proof}

\subsection{Properties of polynomial packings}\label{app:operations}\label{app:bilinear}

\begin{proof}[Proof of Observation~\ref{obs:elementary}]
Use zero insertion and output projection for padding, set discarded
inputs to zero for target deletion, and ignore unused sources.
For SL target deletion, zero only the unwanted right inputs, preserving
the common left input. Independent copies of the maps repeat a PP;
repetition multiplies both counts and both total efforts by $k$.
These operations preserve exactness and polynomial errors divisible
by $\varepsilon$.
\end{proof}

\begin{proof}[Proof of Lemma~\ref{lem:role-exchanges}]
Role exchanges permute exponent-range lengths and preserve $pqr$,
so~\eqref{eq:effort} gives
$\operatorname{eff}(P^\sigma)=\operatorname{eff}(P)$.

The exchange $(12)$ follows from
$(AB)^{\mathsf T}=B^{\mathsf T}A^{\mathsf T}$, by swapping the encoders
and transposing coefficient matrices.

For $(23)$, let
$P=\operatorname{Prod}_{\mathcal I,\mathcal J,\mathcal K}[p,q,r]$,
write $P_A(B)=P(A,B)$, and let $\iota$ invert monomial labels.
Here transposing a linear map means transposing its matrix on
coefficient lists.
For any output list $Z$, extending coefficient lists by zero,
\eqref{eq:selected-product} gives
\[
 [v^{-1}]
 \widehat{P}(A^{\mathsf T},\iota Z)
 =\sum_{u\in\mathcal I}A_u^{\mathsf T}Z_{uv}
 =(P_A^{\mathsf T}Z)_v
 \qquad(v\in\mathcal J).
\]
Thus transposing the packing identity with fixed left inputs gives
\[
 T_A=D S_{E^LA}E^R
 \quad\Longrightarrow\quad
 T_A^{\mathsf T}=(E^R)^{\mathsf T}S_{E^LA}^{\mathsf T}D^{\mathsf T}.
\]
Under these coefficient identifications, the dual encoders and decoder
are $E^L,D^{\mathsf T},(E^R)^{\mathsf T}$. Transposition and reindexing
preserve their allowed coefficients, the counts, and polynomial
errors divisible by $\varepsilon$. Every role permutation is a
composition of $(12)$ and $(23)$, so repeated application proves
the PP conclusion for every $\sigma\in S_3$.

For an SL-PP, fix the common left input $A$ and apply the same
transpose identity to its right-input map. Since every source receives
$A$ unchanged, every dual source receives the same $A^{\mathsf T}$
unchanged. Thus the new left encoder is again the copying map,
and the preceding count and error conclusions also apply.
\end{proof}

\begin{proof}[Proof of Lemma~\ref{lem:composition}]
Write $E_i=(E_i^L,E_i^R)$ for the encoder of packing $i=1,2$. Use
\[
 E=(E_2^L\circ E_1^L,\ E_2^R\circ E_1^R),
 \qquad D=D_1\circ D_2.
\]
Substitution proves exact correctness, with $m$ targets and $\ell$ sources.
Modulo $\varepsilon$, choose nonnegative integers $p_L,p_R,p_D$
such that $\varepsilon^{p_L}E_1^L$, $\varepsilon^{p_R}E_1^R$, and
$\varepsilon^{p_D}D_1$ have entries polynomial in $\varepsilon$.
First replace $\varepsilon$ by
$\varepsilon^{p_L+p_R+p_D+1}$ in $(E_2,D_2)$.
By Lemma~\ref{lem:input-linearity}, substituting the outer encoders
and applying $D_1$ cancels at most $p_L+p_R+p_D$ powers of
$\varepsilon$ from the inner error.
The added error is therefore polynomial and divisible by $\varepsilon$.
The compression rates multiply because the total effort of $(S_j)$
cancels in their product.
\end{proof}

\begin{proof}[Proof of Lemma~\ref{lem:combine}]
Pair the lists, giving $mm'$ targets and $nn'$ sources. For each
encoder and decoder $M\in\{E^L,E^R,D\}$, define its entries by
\[
 M_{(a,a'),(b,b')}=(M_1)_{a,b}(M_2)_{a',b'}.
\]
By Definition~\ref{def:tensor}, expanding the source products and
these maps gives sums that separate over the two sets of indices.
Thus, if $c^{(i)}_{o,a,b}$ is the coefficient of the product of
left-input entry $A_a$ and right-input entry $B_b$ in output
$o$ of packing $i$, the combined coefficient is
\[
 c_{(o,o'),(a,a'),(b,b')}
 =c^{(1)}_{o,a,b}\,c^{(2)}_{o',a',b'}.
\]
Together with Lemma~\ref{lem:input-linearity}, this proves exact
correctness for arbitrary inputs. Modulo $\varepsilon$, these products
remain polynomial with the required constant terms, and their degrees
are at most $d_\varepsilon+d'_\varepsilon$. For SL-PPs, use the common left indices;
the product of the copying encoders copies the combined left input.

Effort multiplies by~\eqref{eq:effort}, since the variable lists are
disjoint and the matrix dimensions multiply. Summing over the paired
lists multiplies the total target efforts and the total source efforts,
so the compression rates multiply as well.
\end{proof}

\begin{proof}[Proof of Lemma~\ref{lem:isomorphism}]
Identity, inverse, and composed variable permutations prove reflexivity,
symmetry, and transitivity. Regrouping or exchanging tensor factors
only regroups or exchanges their disjoint variable lists. Renamings
on separate factors combine into one renaming of their union, proving
compatibility with tensoring. Role exchanges permute the three supports
and invert the appropriate monomials; these operations commute with
variable renaming and distribute over products in disjoint variables.
This proves both assertions involving role exchanges.

A variable renaming induces coefficient permutations on inputs and
outputs. For a target replacement, undo its input permutations before
encoding and apply its output permutation after decoding. For a source
replacement, permute its encoded inputs and undo its output permutation
before decoding. These changes preserve the counts and correctness.
For an SL-PP, one common renaming leaves the left encoder a copying map.
Coefficient permutations also preserve polynomial errors divisible by
$\varepsilon$.
\end{proof}

\begin{proof}[Proof of Lemma~\ref{lem:evaluation}]
Empty supports give the zero product. Otherwise, univariate
interpolation at $a+b-1$ distinct nonzero rational points gives an exact
packing for input intervals of lengths $a,b$: their product occupies
an interval of $a+b-1$ exponents, and multiplying by a monomial reduces
coefficient recovery to ordinary interpolation. Evaluation and recovery
are linear, and any requested output coefficients outside this interval
are zero. For output interval length $h$, apply the same construction
to the $(23)$ and $(13)$ role exchanges, whose input interval lengths
are $(a,h)$ and $(h,b)$, respectively. Exchanging the roles back by
Lemma~\ref{lem:role-exchanges} gives bounds $a+h-1$ and $b+h-1$;
the scalar multiplication sources are unchanged.
Thus the minimum bound is $e_1$.
Tensor these interval packings over the variables by
Lemma~\ref{lem:combine}. To restrict to the actual supports, insert
zeros at missing input monomials and discard unwanted output
coefficients; these linear maps preserve exactness and the source count.
Apply the resulting scalar packing to each triple of matrix indices
and sum over the inner index; this uses
$pqr\prod_k e_k=\operatorname{eff}(P)$ scalar multiplications.
\end{proof}

\subsection{Bini's approximate-to-exact conversion}\label{app:block-recursion}

\begin{proof}[Proof of Lemma~\ref{lem:bini}]
Regard the given bilinear algorithm as a PP with $r$ scalar
multiplication sources, as above.
Let $d_\varepsilon$ bound the degree in $\varepsilon$ of the output
of this packing; the output is polynomial in $\varepsilon$ by the
definition of validity modulo $\varepsilon$. Tensor $h$ copies of the
packing using Lemma~\ref{lem:combine}.
Constant matrix products
remain constant matrix products, so this gives side $v^h$ using
$r^h$ scalar products. Its output has degree at most $hd_\varepsilon$
in $\varepsilon$ and the required matrix product as the constant term.
Evaluate the rule at $hd_\varepsilon+1$ distinct nonzero rational
values of $\varepsilon$ and interpolate its constant coefficient.
Evaluation and interpolation use rational linear maps, so this gives
an exact rational bilinear algorithm using at most
$(hd_\varepsilon+1)r^h$ scalar multiplications.
Since
\[
 (hd_\varepsilon+1)\left(\frac{r}{v^\tau}\right)^h\longrightarrow0,
\]
some finite $h$ gives fewer than $V^\tau$ scalar multiplications
for $V=v^h\ge2$. All steps are computable from the given algorithm.
\end{proof}

\end{document}